%% file: main.tex
\documentclass[letterpaper, 10 pt, conference]{ieeeconf}

\input{preamble}

\title{{\color{blue}Safe and Non-Conservative Obstacle Avoidance under State Estimation Error} with Learning-Based Measurement-Robust Control Barrier Functions}
\title{Learning-Based Measurement-Robust Control Barrier Functions for\\ Obstacle Avoidance under State Estimation Error}
\author{Anonymous Authors}
\iftrue
\author{Nicholas Rober, Yixuan Jia, and Jonathan P.\ How%
\thanks{Submitted 01/12/2026. Research supported by Aurora Flight Sciences.}%
\thanks{N.\ Rober, Y.\ Jia, and J.\ How are with the Aerospace Controls Lab, Department of Aeronautics and Astronautics, Massachusetts Institute of Technology, Cambridge, MA 02319 USA (e-mail: nrober@mit.edu, ameredit@mit.edu, jhow@mit.edu).}%
}
\fi
\begin{document}

\newif\ifarxiv
\arxivfalse

\maketitle
\thispagestyle{empty}
\iftrue
\setlength{\skip\footins}{8pt}
\begingroup
\renewcommand{\thefootnote}{}
\footnotetext{Submitted August $20^\mathrm{th}$, 2026.
This material is based upon work supported by the Naval Information Warfare Center (NIWC) Atlantic under Contract No. N6523623C8011.
The views, opinions, and/or findings expressed are those of the author(s) and should not be interpreted as representing the official views or policies of the Department of Defense or the U.S. Government.}
% \footnotetext{Code: \url{https://github.com/mit-acl/guardian}}
% \vspace{-6pt}
\endgroup
\fi

\begin{abstract}
Safety filters are an effective tool for enforcing constraints in safety-critical systems, but most existing methods assume perfect state information, which is rarely available in practice.
Recent work has begun to close this gap by developing filtering mechanisms that are robust to state estimation error, but these methods can still exhibit safety violations or overly conservative behavior as estimation error grows. Focusing on obstacle avoidance, we develop two new control barrier function (CBF) formulations: drift-measurement-robust (DMR)-CBFs and neural measurement-robust (NMR)-CBFs. The DMR-CBF augments the standard CBF condition with an inner optimization over the worst-case uncertainty in the drift dynamics, improving robustness to estimation error. This DMR-CBF then supervises a pretraining phase for the NMR-CBF, which replaces the inner optimization with a learned term. The NMR-CBF is subsequently finetuned through differentiable trajectory rollouts, 
yielding a filter that achieves empirical safety comparable to the DMR-CBF while reducing both conservativeness and computational cost.
We provide theoretical analysis of the DMR-CBF along with numerical results on a planar double integrator and a 12D quadrotor, where both proposed approaches prevent collisions while other robust methods either fail or are overly conservative. Finally, we deployed the NMR-CBF on a Unitree Go2, enabling successful navigation of an obstacle field under odometry errors that caused a standard CBF to collide.
\end{abstract}

\input{introduction}
\input{preliminaries}
\input{approach}
\input{results}
\input{conclusion}

\bibliographystyle{aiaa}
% \balance
\bibliography{refs}
\end{document}

%% file: preamble.tex
\usepackage[utf8]{inputenc}
\usepackage{amsmath}
\usepackage{amssymb}
\usepackage{graphicx,balance}
\usepackage{pifont}
\usepackage{bm}
\usepackage{dsfont}
\usepackage{subcaption}
\usepackage{booktabs}
\usepackage{multirow}
\usepackage{makecell}
\usepackage{siunitx}

\usepackage{amsfonts}
\usepackage{url}
\usepackage[pdftex,plainpages = false, colorlinks=true, linkcolor=black, citecolor = black, urlcolor = blue,pagebackref=false,hypertexnames=false, plainpages=false, pdfpagelabels]{hyperref}
\usepackage{balance}
\usepackage[capitalize]{cleveref}
\usepackage[sort,compress]{cite}

\newtheorem{theorem}{Theorem}[section]

\newtheorem{proposition}[theorem]{Proposition}
\crefname{section}{Section}{Sections}
\crefname{theorem}{Theorem}{Theorems}
\crefname{lemma}{Lemma}{Lemmas}
\crefname{table}{Table}{Tables}
\crefformat{equation}{(#2#1#3)}
\crefname{algocf}{Algorithm}{Algorithms}
\Crefname{algocf}{Algorithm}{Algorithms}
\crefname{ALC@unique}{Line}{Lines}

\newtheorem{lem}{Lemma}

\newtheorem{prop}{Proposition}

\crefname{asm}{Assumption}{Assumptions}
\Crefname{asm}{Assumption}{Assumptions}

\newcommand{\x}{\bm{x}}
\newcommand{\s}{\bm{s}}

\renewcommand{\u}{\bm{u}}
\newcommand{\errs}{\mathcal{E}_{\hat{\x}}}
\newcommand{\uncs}{\bar{\mathcal{X}}}

\newcommand{\y}{y}
\newcommand{\e}{\bm{e}}

\newcommand{\p}{\bm{p}}
\renewcommand{\v}{\bm{v}}

\usepackage{accents}

\usepackage{paralist}

\crefformat{chapter}{\S#2#1#3}
\crefmultiformat{chapter}{\S\S#2#1#3}{and~#2#1#3}{, #2#1#3}{, and~#2#1#3}
\crefformat{section}{\S#2#1#3}
\crefmultiformat{section}{\S\S#2#1#3}{and~#2#1#3}{, #2#1#3}{, and~#2#1#3}

\usepackage{algorithm,algorithmic}

\usepackage{tikz}
\usetikzlibrary{shapes,arrows,arrows.meta,positioning,calc}

\newif\ifshrink
\shrinktrue
\newcommand{\gdn}{GUARDIAN}
\definecolor{review}{RGB}{0, 0, 0}
\definecolor{internal}{RGB}{0, 0, 0}

\newcommand{\ana}{DMR-CBF}

%% file: introduction.tex
\section{Introduction}
\label{sec:introduction}
% The modified HJ reachability formulation presented in \cref{chap:guardian} provides strong guarantees for NFLs with learned perception modules, i.e., ONFLs.
% However, as described in \cref{guardian:sec:approach}, these guarantees are dependent on a set of assumptions that require there to be a uniformly safe control direction over the entire state uncertainty set.
% As will be demonstrated in this chapter, this limitation is not strictly theoretical; \gdn\ allows real safety violations when its assumptions are not satisfied.
% % Considering the results in \cref{guardian:sec:comparisons} demonstrating that existing methods, though not reliant on the same assumptions as \gdn, struggle under the high state estimation error regime that may arise as a result of adversarial attacks.
% The results in \cref{guardian:sec:comparisons} further demonstrate that existing methods, though not reliant on the same assumptions as \gdn, struggle under the high state estimation error regime that may arise as a result of adversarial attacks.
% This leaves safety for ONFLs in situations like obstacle avoidance, where \gdn's assumptions do not hold, as an open area of research.

As autonomous systems become increasingly prevalent in safety-critical settings such as automotive and aerospace applications~\cite{bojarski2016end,kaufmann2023champion}, ensuring that they respect the constraints of their environments has become essential to protecting the systems themselves, the people around them, and their surroundings.
% As autonomous systems see growing use in safety-critical settings such as automotive, medical, and aerospace applications~\cite{}, ensuring that they respect the constraints of their environments has become essential to protecting the systems themselves, the people around them, and their surroundings.
Safety filtering is a powerful tool for this task, modifying the commands of a potentially unsafe nominal controller to prevent violations of the system's safety constraints.
Common safety filtering approaches include control barrier functions (CBFs)~\cite{ames2016control} and Hamilton-Jacobi (HJ) reachability analysis~\cite{bansal_hamilton-jacobi_2017}, which differ in their underlying mechanisms but share the goal of minimally adjusting an unsafe nominal input so that the filtered input satisfies the safety constraints.
While effective in their intended settings, most safety filters assume access to perfect state information, which is rarely available in practice.
Real systems instead rely on a state estimate produced by an estimation module, such as a Kalman filter~\cite{kalman1960new} or a learned perception model~\cite{katz_verification_2022}.
Such estimates inevitably carry error, which can be substantial in the case of learned models susceptible to adversarial attacks~\cite{szegedy_intriguing_2014}.
Thus, there is a need for filtering mechanisms that preserve safety despite (potentially large) state estimation error, such as the drift between the true and estimated state visible in \cref{fig:go2_nmr_composite}.

\begin{figure}[tp]
    \centering
    \includegraphics[width=1.0\columnwidth]{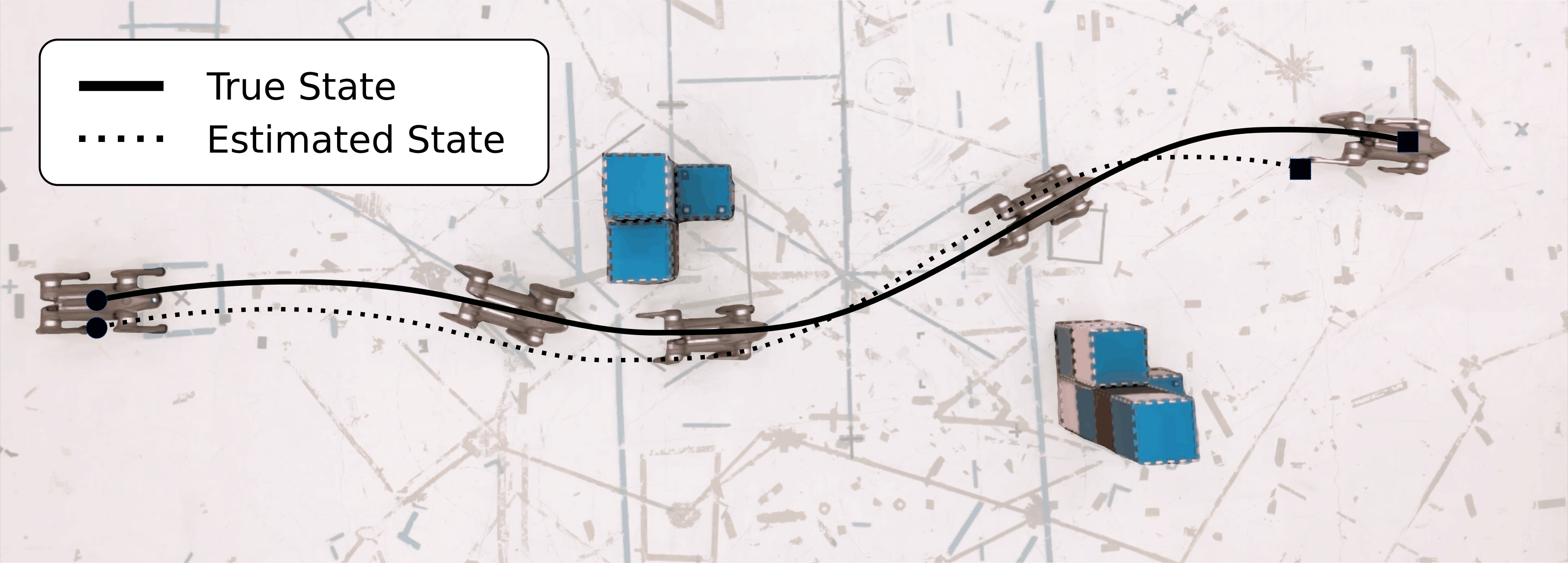}
    \caption{Our approach (NMR-CBF) enables safe and performant obstacle avoidance despite state estimation error.}
    \label{fig:go2_nmr_composite}
    \ifshrink
    \vspace{-16pt}
    \fi
\end{figure}

In response to this issue, several recent works have developed safety filters that are robust to state estimation error~\cite{tan2024safety,tan2025secure,agrawal_safe_2023,dean_guaranteeing_2021,nanayakkara_safety_2025,das_safe_2025,tan2026duality,rober2026guardian}.
Early efforts addressed sensor-spoofing attacks on linear systems~\cite{tan2024safety,tan2025secure}, and \cite{agrawal_safe_2023} developed a filtering mechanism for control-affine nonlinear systems whose estimation modules have a known dynamics model.
More generally, for systems with only a bounded estimation error, a variety of methods have emerged that share a common structure despite their differing mechanisms: each modifies the nominal safety condition to account for the discrepancy between the true state and the estimate, and how they do so separates them into two broad paradigms.
\textit{Point-based} methods (MR-CBFs~\cite{dean_guaranteeing_2021}, R-CBFs~\cite{nanayakkara_safety_2025}, and R-CBF-QPs~\cite{das_safe_2025}) evaluate the safety condition at the state estimate and absorb uncertainty through an added robustness term that tightens the constraint.
\textit{Set-based} methods (the Duality CBF~\cite{tan2026duality} and \gdn~\cite{rober2026guardian}) instead find a control that is safe over the entire state uncertainty set.
As detailed in \cref{sec:prelims:prev}, each paradigm carries a characteristic failure mode in obstacle avoidance.
Point-based methods tend to permit collisions as estimation error grows, while set-based methods become overly conservative when it is ambiguous which way to steer around an obstacle.
% Point-based methods tend to permit collisions as state estimation error grows and set-based methods tend to be overly conservative when there is uncertainty about which direction to go around the obstacle.
Set-based methods additionally scale poorly with state dimension, limiting their application to smaller systems.

To address these limitations, we propose two CBF formulations that combine elements of point- and set-based robustness. The DMR-CBF provides a certifiable formulation, while the learned NMR-CBF reduces its conservativeness and computational cost through supervised pretraining and differentiable trajectory fine-tuning. 

The contributions of this paper are therefore as follows:
\begin{itemize}
    \item {\color{internal}\ana: a CBF formulation with an \emph{a posteriori} certificate that protects systems subject to estimation error without the conservativeness of set-based filters or the collision-proneness of point-based filters}.
    % \item NMR-CBF: a CBF with a learned robustness term, pretrained on \ana-generated data and fine-tuned via differentiable trajectory rollouts, that matches or beats the task performance of the baselines while allowing the fewest safety violations.
    %\item NMR-CBF: a CBF formulation with a learned robustness term reducing conservativeness and computational cost relative to the \ana\ while {\color{magenta} empirically} preserving its safety.
    \item NMR-CBF: a CBF formulation with a learned robustness term reducing conservativeness and computational cost relative to the \ana\ while achieving comparable empirical safety. 
    % {\color{red} ok with this wording? seems to be better aligned with what we show}
    %\item NMR-CBF: a CBF formulation with a learned robustness term reducing conservativeness and computational cost relative to the \ana\ while preserving safety.
    % \item Numerical experiments demonstrating how our methods improve [hmmmm] baselines across obstacle-avoidance scenarios and varying levels of state estimation error.
    \item Numerical experiments across obstacle-avoidance scenarios showing our methods achieve both safety and performance whereas the baselines are either safe but conservative or performant but collision-prone.
    \item Hardware experiments demonstrating the NMR-CBF (\cref{fig:go2_nmr_composite}) on a Unitree Go2 under onboard estimate drift that causes a standard CBF to collide with an obstacle.
    %\item Hardware experiments where our NMR-CBF (\cref{fig:go2_nmr_composite}) protects a Unitree Go2 against onboard estimate drift that causes a standard CBF to collide with an obstacle.
    % \item \ana: a theoretically certifiable CBF formulation that overcomes the limitations of the set- and point-based measurement-robust safety filters.
    % \item NMR-CBF: a CBF formulation with a learned robustness term that is trained to be less conservative than the \ana, matching or beating the performance of all baselines while allowing the fewest safety violations.
    % \item Extensive numerical experiments comparing the proposed methods with the baselines demonstrate our approaches allow the fewest safety violations.
    % \item Hardware experiments demonstrate our NMR-CBF protecting a Unitree Go2 against onboard measurement drift that causes a standard CBF to crash into an obstacle.
\end{itemize}

%% file: preliminaries.tex
\section{Preliminaries}
\label{sec:prelims}
\textit{Notation:} The $L_2$ and $L_\infty$ norms of $\x\in\mathbb{R}^n$ are denoted as $\|\x\|$ and $\|\x\|_\infty$, respectively.
Unless otherwise specified, inequalities and absolute values operate element-wise for vectors, i.e., for $\x, \y \in \mathbb{R}^n$, $|\x| \leq \y\ \implies |\x_i| \leq \y_i\ \forall i\in\{1\ldots n\}$.
For functions $f: \mathbb{R}^n \rightarrow \mathbb{R}^n$ and $h: \mathbb{R}^n \rightarrow \mathbb{R}$ and $\x\in\mathbb{R}^n$, $L_fh(\x)$ denotes the Lie derivative of $h$ along $f$ at $\x$, i.e., $\nabla h(\x) f(\x)$.
A function $\alpha: \mathbb{R} \rightarrow \mathbb{R}$ is an extended class $\mathcal{K}$ function if it is strictly monotonically increasing and $\alpha(0) = 0$. 
For two sets, $\mathcal{X}, \mathcal{Y} \subseteq \mathbb{R}^n$, $\mathcal{X} \oplus \mathcal{Y}$ denotes their Minkowski sum.

\subsection{System Dynamics}
\label{sec:prelims:dynamics}
Consider the nonlinear system with state vector $\x \in \mathcal{X} \subseteq \mathbb{R}^{n}$ and control-affine dynamics
\begin{equation}
\label{eqn:dynamics}
    \dot{\x} = f(\x) +  g(\x) \u
\end{equation}
where the drift dynamics $f(\x):\mathbb{R}^{n} \rightarrow \mathbb{R}^{n}$ and input matrix $g(\x):\mathbb{R}^{n} \rightarrow \mathbb{R}^{n\times m}$ are locally Lipschitz continuous and $\u \in \mathcal{U} \subset \mathbb{R}^{m}$ is the control input.
The true state of the system is unknown, but we assume there is a known state estimate $\hat{\x}$ with a known error bound $\e_{\hat{x}}: \mathcal{X} \times \mathbb{R}_{\geq0} \rightarrow \mathbb{R}_{\geq0}^{n}$, i.e., $|\hat{\x} - \x| \leq \e_{\hat{x}}$.
Such a bound can be obtained via, e.g, known estimator properties or verification tools as described in \cite{rober2026guardian}, and can be used to construct the error set $\errs \triangleq \{\s \in \mathbb{R}^n \ | \ |\s| \leq \e_{\hat{x}}\}$ and the state uncertainty set $\uncs \triangleq \{\hat{\x}\} \oplus \errs$ centered at $\hat{\x}$ with element-wise radii defined by $\e_{\hat{x}}$.
% The system is subject to safety constraints that can be encoded as 
% The system has a Lipschitz continuous controller $\u = k_d(\hat{\x})$, so the closed loop system takes the form $\dot{\x} = f(\x) +  g(\x) k_d(\hat{\x})$.

\subsection{Control Barrier Functions~\cite{ames2016control}}
\label{sec:prelims:cbfs}
To formally define safety in the context of CBFs, we consider a set $\mathcal{C} = \{\x \in \mathcal{X} \ | \ h(\x) \geq 0 \}$, where $h(\x): \mathcal{X} \rightarrow \mathbb{R}$ is a continuously differentiable function satisfying $\frac{d h}{d \x} \not= 0$ when $h(\x)=0$ and where $\mathcal{C}$ is non-empty and has no isolated points.
Safety can then be considered with respect to $\mathcal{C}$ via a forward invariance condition: if an initial state $\x_0 \triangleq \x(0) \in \mathcal{C}$, then $\forall t > 0$, the solution to \cref{eqn:dynamics} at time $t$, i.e.,  $\x(t)$, satisfies $\x(t) \in \mathcal{C}$, indicating that the system is safe with respect to $\mathcal{C}$.
The function $h$ is then considered a control barrier function (CBF) if there exists an extended class $\mathcal{K}$ function $\alpha(\cdot)$ such that 
\begin{equation}
\label{eqn:cbf_condition}
    \sup_{\u \in \mathcal{U}} L_fh(\x) + L_gh(\x)\u + \alpha(h(\x)) \geq 0.
\end{equation}
In the setting where $\x$ is known, a safe controller can then be constructed as a quadratic program (QP)
\begin{align}
\label{eqn:cbf}
% \tag{CBF}
    k_{\mathrm{cbf}}(\x) & = \mathrm{arg}\min_{\u \in \mathcal{U}} \| \u - k_d(\x) \|^2 \\
    & \mathrm{s.t.}\ L_fh(\x) + L_gh(\x)\u + \alpha(h(\x)) \geq 0 \label{eqn:cbf_constraint}
\end{align}
that minimizes intervention against a locally Lipschitz nominal controller $k_d : \mathcal{X} \rightarrow \mathcal{U}$ while ensuring the satisfaction of \cref{eqn:cbf_condition}.
For the rest of this paper, we consider the case where the state estimate $\hat{\x}$ is used in place of the unknown true state $\x$, but \cref{eqn:cbf_condition} must still be maintained.

\subsection{Safety Filtering Under State Estimation Error}
\label{sec:prelims:prev}
% Safety filtering is a more general class of formulations  
Given that the true state $\x$ is rarely available and we must typically rely on the state estimate $\hat{\x}$, several recent works have addressed the problem of safety filtering under state estimation error.
GUARDIAN~\cite{rober2026guardian} is unique in its development of a modified Hamilton-Jacobi (HJ) reachability formulation, but the primary thread has focused on CBF-based formulations that solve the optimization \cref{eqn:cbf} but use a modified constraint \cref{eqn:cbf_constraint}.
The most prevalent strategy, referred to here as a \textit{point-based} robust strategy, is to formulate a constraint that is dependent on the state estimate, but includes an additional tightening term $\rho(\cdot)$, so that \cref{eqn:cbf_constraint} is replaced by 
\begin{equation}
    L_fh(\hat{\x}) + L_gh(\hat{\x})\u + \alpha(h(\hat{\x})) - \rho(\cdot) \geq 0.
\end{equation}
Examples of safety filters that use a constraint of this form include Measurement-Robust (MR)-CBFs~\cite{dean_guaranteeing_2021}, Robust (R)-CBFs~\cite{nanayakkara_safety_2025} and R-CBF-QPs~\cite{das_safe_2025}, each of which include different robustifying terms $\rho(\cdot)$.
An alternative strategy recently proposed by \cite{tan2026duality,rober2026guardian} is to instead use a \textit{set-based} robust strategy, replacing \cref{eqn:cbf_constraint} with something of the form
\begin{equation}
    L_fh(\s) + L_gh(\s)\u + \alpha(h(\s)) \geq 0\quad \forall \s \in \uncs.
\end{equation}
Duality CBFs~\cite{tan2026duality} accomplish this using a dual formulation of \cref{eqn:cbf} and developing a corresponding constraint, and \gdn~\cite{rober2026guardian} employs a philosophically similar approach using HJ reachability.

Each of these formulations suffers from limitations that align with its categorization as either point-based or set-based.
% Each of these existing formulations suffer from issues that manifest in similar ways by categorization as either point-based or set-based.
First considering the point-based approaches, the MR-CBF provides strong guarantees for the original safe set $\mathcal{C}$, but the existence of a standard CBF does not guarantee the existence of a valid MR-CBF, and they may be difficult or impossible to find, especially as the state-estimation error increases.
Alternatively, R-CBFs can always be formed from an existing CBF, but rely on two hyper-parameters, $\gamma_1,\ \gamma_2\in \mathbb{R}_{>0}$ that do not adapt to varying levels of state estimation error, leading to possible safety violations.
R-CBF-QPs adaptively tune $\gamma_1$ and $\gamma_2$, but, past a certain level of state estimation error, are only guaranteed safe for an inflated version of $\mathcal{C}$.

Regarding the set-based approaches, the Duality CBF~\cite{tan2026duality} is guaranteed safe if its dual optimization is well-formulated, but its need to find a control that uniformly preserves safety over the entire state uncertainty set can result in overly conservative obstacle avoidance, as will be shown in \cref{sec:results}.
\gdn\ similarly requires a uniform safe control direction over the state uncertainty set, but additionally requires a scalar control input, so it is not a valid solution to many obstacle avoidance problems.
Moreover, as demonstrated in \cref{sec:results}, both the Duality CBF and \gdn\ approaches suffer from scalability issues, limiting their practical application to smaller systems when compared with the point-based filtering approaches.
\iffalse
Thus, the distinction between point-based and set-based robust safety filtering mechanisms has real implications for their respective performance in obstacle avoidance: point-based approaches may violate safety, especially as state estimation error increases, and set-based approaches are overly conservative and scale worse with state dimension.
The problem addressed in this chapter is then to find a filtering mechanism that resolves these issues by allowing a system to safely avoid obstacles in the presence of large estimation errors while not being overly conservative.
\fi
Thus, we seek a filter that maintains safety under large estimation errors without the conservativeness and scalability limitations of set-based approaches.

%% file: approach.tex
\section{Approach: \ana s and NMR-CBFs}
\label{sec:approach}

%%%%%% Need to verify correctness - claude helped construct this section %%%%%%
%=====================================================================
%  PSHR-CBF: Theory section
%  ------------------------------------------------------------------
%  Assumed theorem environments (declare in preamble if not present):
%     \newtheorem{lemma}{Lemma}
%     \newtheorem{proposition}{Proposition}
%     \newtheorem{theorem}{Theorem}
%     \newtheorem{corollary}{Corollary}
%     \newtheorem{definition}{Definition}
%     \newtheorem{remark}{Remark}
%  (proof / cref / bm / amsmath assumed available)
%
%  Notation hooks into your existing labels:
%     \cref{eqn:cbf_condition}, \cref{eqn:cbf}, \cref{eqn:duality-cbf},
%     \cref{ex:dubins}, \cref{fig:point_vs_set}, \cref{sec:results}.
%  Uses your macros: \x \s \u \uncs \errs \hat{\x}.
%=====================================================================

% \label{nmrcbfss:sec:dmr}
In this section, we propose two CBF formulations for obstacle avoidance under state estimation uncertainty. \Cref{sec:approach:dmrcbfs} introduces the {drift–measurement-robust} CBF (\ana), which finds the worst-case drift term in the standard CBF condition~\cref{eqn:cbf_constraint} and uses it to find a safe control. Building on the \ana, \cref{sec:approach:nmrcbf} introduces the {neural measurement-robust} CBF (NMR-CBF), which learns to imitate the \ana\ during pretraining and is subsequently fine-tuned using trajectory rollouts. As shown in \cref{sec:results} and \cref{sec:hardware}, the resulting filter is less conservative and empirically safe, despite lacking the guarantees of its \ana\ teacher.

\subsection{Drift--Measurement Robust CBFs}
\label{sec:approach:dmrcbfs}
The \ana\ uses the modified constraint
\begin{equation}
\label{eqn:dmr}
% \tag{PSHR-CBF}
    \min_{\s \in \uncs}\big[L_fh(\s) + \alpha(h(\s))\big] + L_gh(\hat{\x})\,\u \;\geq\; 0 .
\end{equation}
Conceptually, {\color{internal}the \ana\ strikes a balance between the set- and point-based robust methods: } it is set-robust in the autonomous (drift) term $L_fh(\cdot) + \alpha(h(\cdot))$, but {\color{internal} point-robust} in the control term $L_gh$.
% By taking the worst case over the uncertainty set $\uncs$ in everything that does not depend on $\u$, while evaluating the control coefficient $L_gh$ at the estimate $\hat{\x}$ alone
% Conceptually, the \ana\ is set-robust in the autonomous (drift) term $L_fh(\cdot) + \alpha(h(\cdot))$, but nominal in the control term $L_gh$: it takes the worst case over the uncertainty set $\uncs$ in everything that does not depend on $\u$, while evaluating the control coefficient $L_gh$ at the estimate $\hat{\x}$ alone.
%{\color{internal}By taking the worst-case drift term over the uncertainty set $\uncs$, the \ana\ incorporates a measure of robustness similar to the state-of-the-art Duality CBF~\cite{tan2026duality}.
%However, by evaluating the control coefficient $L_gh$ at the estimate $\hat{\x}$ alone, it avoids issues associated with determining a control that is uniformly safe over all of $\uncs$. This allows for the \ana\ to be robust to estimation error while avoiding the indecisiveness common to set-based approaches, as will be shown in \cref{sec:results}.
%}
By taking the worst-case drift term over the uncertainty set $\uncs$, the \ana\ accounts for uncertainty in the uncontrolled evolution of the safety margin, similarly to the state-of-the-art Duality CBF~\cite{tan2026duality}. In contrast, requiring the control term to be robust over all of $\uncs$ can force a single control action to accommodate conflicting control directions associated with different possible states, leading to the indecisiveness characteristic of set-based approaches. The \ana\ therefore evaluates $L_gh$ only at the estimate $\hat{\x}$, retaining nominal control authority while robustifying the portion of the CBF condition that is independent of the chosen control. This construction seeks to balance robustness to estimation error against the conservativeness of fully set-based control, as demonstrated in \cref{sec:results}.

% {\color{red} seems like we need to commentary here about why this is a good idea}
% This subsection makes that statement precise, characterizes exactly when \cref{eqn:dmr} preserves safety of the true state, shows that the residual failure mode is structurally confined to a thin set, and turns that characterization into a certified controller.

To formally analyze the properties of the \ana, we first define the autonomous and control terms as
\begin{equation}
\label{eqn:ab_def}
    a(\s) \triangleq L_fh(\s) + \alpha(h(\s)),
    \quad
    b(\s) \triangleq L_gh(\s) \in \mathbb{R}^{1\times m},
\end{equation}
so that the true CBF condition required along the real trajectory is $a(\x) + b(\x)\u \geq 0$.
We define the worst-case autonomous value as
\begin{equation}
\label{eqn:amin_def}
    \underline{a} \triangleq \min_{\s\in\uncs} a(\s)
    = \min_{\s\in\uncs}\big[L_fh(\s) + \alpha(h(\s))\big].
\end{equation}
With this notation, \cref{eqn:dmr} reads $\underline{a} + b(\hat{\x})\u \geq 0$.
Recall $\hat{\x}\in\uncs$ and that $\uncs = \{\hat{\x}\} \oplus \errs$ is a box, hence convex and connected, so the true state satisfies $\x\in\uncs$ by construction of the error bound.
% All results extend to HOCBFs by substituting $\psi_{r-1}$ for $h$, exactly as in \cref{sec:prelims:hocbfs}.

The true CBF condition admits the following decomposition:
%Given that we want to ensure the CBF condition is true for the true state, i.e., $a(\x) + b(\x)\u \geq 0$, we rely on the following decomposition:
\begin{lem}[Decomposition]
\label{lem:decomp}
For every $\u\in\mathcal{U}$ and every $\x\in\mathcal{X}$,
\begin{equation}
\label{eqn:decomp}
\begin{split}
    a(\x) + & b(\x)\u
    = \\ & \underbrace{\big[\underline{a} + b(\hat{\x})\u\big]}_{\textnormal{DMR slack}}
    + \underbrace{\big[a(\x) - \underline{a}\big]}_{\textnormal{drift gap}}
    + \underbrace{\big[b(\x) - b(\hat{\x})\big]\u}_{\Delta_g(\x,\u)}.
\end{split}
\end{equation}
\end{lem}

\begin{proof}
Expand the right-hand side: the $\underline{a}$ and $b(\hat{\x})\u$ terms cancel between the first and remaining brackets, leaving $a(\x) + b(\x)\u$.
\end{proof}

{\color{internal}This decomposition reveals an important safety property of the \ana, stated formally in \cref{cor:pointwise}.}

\begin{lem}[Pointwise lower bound]
\label{cor:pointwise}
If $\u$ satisfies the \ana\ constraint \cref{eqn:dmr} and $\x\in\uncs$, then
\begin{equation}
\label{eqn:pointwise}
    a(\x) + b(\x)\u \;\geq\; \Delta_g(\x,\u) = \big[L_gh(\x) - L_gh(\hat{\x})\big]\u.
\end{equation}
\vspace{-0.8\baselineskip}
\end{lem}

\begin{proof}
In \cref{eqn:decomp} the DMR slack is nonnegative by \cref{eqn:dmr}, and the drift gap is nonnegative because $a(\x)\geq\underline{a}$ for any $\x\in\uncs$ by \cref{eqn:amin_def}. The remaining term is $\Delta_g(\x,\u)$.
\end{proof}

{\color{internal}Critically, \cref{cor:pointwise} isolates the lone potential failure point of the \ana: by constructing a constraint where the drift term is robust to state estimation error, the control-coefficient mismatch $\Delta_g(\x,\u)$ is the only term that can be negative.}
If $\Delta_g(\x,\u)\geq0$ then the true CBF condition holds and the true state is safe; any violation must arise from a negative $\Delta_g(\x,\u)$.
With this in mind, let
\begin{equation}
\label{eqn:M_def}
    M(\u) \triangleq \min_{\s\in\uncs}\big[a(\s) + b(\s)\u\big]
\end{equation}
denote the {true robust margin}, and let
\begin{equation}
\label{eqn:G_def}
    G(\u) \triangleq b(\hat{\x})\u - \min_{\s\in\uncs} b(\s)\u \;\geq\; 0
\end{equation}
denote the {control-authority gap}, where the inequality $G(\u)\geq0$ holds because $\hat{\x}\in\uncs$.
%: how much more authority the estimate $\hat{\x}$ claims than the worst-case state in $\uncs$ actually grants.

\begin{prop}[\ana\ Certificate]
\label{prop:cert}
For all $\u\in\mathcal{U}$,
\begin{equation}
\label{eqn:cert_bound}
    M(\u) \;\geq\; \big[\underline{a} + b(\hat{\x})\u\big] - G(\u).
\end{equation}
Consequently, if a \ana\ solution $\u^\star$ to \cref{eqn:cbf} (with constraint \cref{eqn:dmr}) satisfies $\underline{a} + b(\hat{\x})\u^\star \geq G(\u^\star)$, then $M(\u^\star)\geq0$, the CBF condition \cref{eqn:cbf_condition} holds at the true state $\x$, and the true safe set $\mathcal{C}$ is rendered forward invariant.
\end{prop}

\begin{proof}
Since the minimum of a sum dominates the sum of minima, $M(\u) = \min_{\s}[a(\s)+b(\s)\u] \geq \min_{\s}a(\s) + \min_{\s}b(\s)\u = \underline{a} + \big(b(\hat{\x})\u - G(\u)\big)$, which is \cref{eqn:cert_bound}.
If the stated inequality holds at $\u^\star$ then the right-hand side of \cref{eqn:cert_bound} is nonnegative, so $M(\u^\star)\geq0$, i.e.\ $a(\s)+b(\s)\u^\star\geq0$ for all $\s\in\uncs$. As $\x\in\uncs$, this holds at the true state, which is \cref{eqn:cbf_condition}; forward invariance of $\mathcal{C}$ follows by the standard CBF argument~\cite{ames2016control}.
\end{proof}
Like $\underline{a}$, $G(\u^\star)$ is a scalar minimization of a fixed function over $\uncs$, evaluated at the solved $\u^\star$, so the \ana\ provides a cheap, \emph{a posteriori} certificate.

\subsection{Neural Measurement-Robust CBFs}
\label{sec:approach:nmrcbf}

The \ana\ requires an inner minimization and may be overly conservative due to its worst-case bound.
The NMR-CBF addresses both limitations by replacing the analytic drift correction $a(\hat\x) - \underline{a}$ with a learned, non-negative residual $\rho_\xi(\hat\x, \e_{\hat\x})$:
\begin{equation}
\label{eqn:nmr}
    L_fh(\hat\x) + L_gh(\hat\x)\,\u + \alpha(h(\hat\x))
    \;-\;\rho_\xi(\hat\x, \e_{\hat\x})
    \;\geq\; 0,
\end{equation}
where $\xi$ denotes the parameters of a feedforward network whose input is the concatenation $[\hat\x;\,\e_{\hat\x}] \in \mathbb{R}^{2n}$ and whose output is passed through a softplus head to enforce $\rho_\xi \geq 0$ for every constraint.
Conditioning on $\e_{\hat\x}$ allows a single trained network to generalize across any per-dimension uncertainty radius up to a chosen training envelope $\overline{\e}$, rather than requiring a separate model per noise level.
% This allows the NMR-CBF to handle varying sensitivity to adversarial attack, discussed in \cref{guardian:sec:numerical_results}.
While the NMR-CBF inherits no formal safety guarantee, as will be shown in \cref{sec:results} and \cref{sec:hardware}, it still demonstrates strong empirical safety performance.

% Like the \ana, the NMR-CBF tightens only the drift term and evaluates $L_gh$ at $\hat\x$; the residual control-coefficient mismatch $\Delta_g(\x,\u) = [L_gh(\x) - L_gh(\hat\x)]\,\u$ analyzed in \cref{lem:decomp} therefore applies unchanged.
% When $\rho_\xi(\hat\x, \e_{\hat\x}) \geq a(\hat\x) - \underline{a}$ pointwise the NMR constraint \cref{eqn:nmr} dominates the \ana\ constraint \cref{eqn:dmr}, and the certificate of \cref{prop:cert} transfers verbatim.
% In general $\rho_\xi$ may under- or over-approximate $a(\hat\x) - \underline{a}$, so the NMR-CBF inherits no formal safety guarantee on its own; safety is enforced empirically via the training loss.
% However, as will be shown in \cref{sec:results,sec:hardware}, it still demonstrates strong empirical safety performance.
%and, when desired, can be made a posteriori certified by composing NMR-CBF with the fallback construction of \cref{sec:dmr:hybrid} (substituting $\u^\star$ from \cref{eqn:nmr} for the \ana\ solution).

\subsubsection{Supervised Pretraining}
\label{sec:approach:nmrcbf:pretrain}

We initialize $\rho_\xi$ by regression against the \ana\ oracle.
For each constraint $i$, define
\begin{align}
\label{eqn:phi-oracle}
    \phi^\star_i(\hat\x, \e_{\hat\x})
    \triangleq \bigl(&L_fh_i(\hat\x) + \alpha(h_i(\hat\x))\bigr) \nonumber \\ 
    &-\bigl[\min_{\s\in\uncs}\bigl(L_fh_i(\s) + \alpha(h_i(\s))\bigr)\bigr],
\end{align}
clipped to $[0, \phi_{\max}]$ to guard against ill-conditioned outliers near CBF singularities.
Concretely $\phi^\star$ equals $a_i(\hat\x) - \underline{a}_i$ from \cref{sec:approach:dmrcbfs} and is computed via projected-gradient ascent with several random restarts over $\uncs$.
We then sample state estimates $\hat\x^j$ uniformly from the state space (rejecting those inside the obstacle), pair each with a per-dimension noise radius $\e^j \sim \mathrm{Uniform}([0,\overline\e])$, and minimize the per-constraint mean-squared error
\begin{equation}
\label{eqn:loss-pretrain}
    \mathcal L_{\mathrm{sup}}(\xi)
    \;=\;
    \frac{1}{N}\sum_{j=1}^{N} \bigl\|
    \rho_\xi(\hat\x^j, \e^j) - \phi^\star(\hat\x^j, \e^j)
    \bigr\|_2^2.
\end{equation}
The pretrained network already serves as a drop-in replacement for the \ana\ inner minimization, at reduced cost set by a single MLP evaluation per step.

\subsubsection{Differentiable Fine-Tuning}
\label{sec:approach:nmrcbf:finetune}

Supervised pretraining incorporates the \ana's worst-case assumption into the NMR-CBF, so it initially inherits its conservativeness.
%: $\phi^\star$ takes the worst $\s\in\uncs$ regardless of whether that worst case is dynamically reachable along the closed-loop trajectory.
We therefore fine-tune $\xi$ on differentiable closed-loop rollouts that allow $\rho_\xi$ to shrink wherever the trajectory has safety slack.%, while penalizing actual safety violations.

Given an initial state $\x_0$ sampled from the operating envelope, an error bound $\e \sim \mathrm{Uniform}\bigl([0,\overline{\e}]\bigr)$, and a bias term $\mathbf{b} \sim \mathrm{Uniform}\bigl([-\e, \e]\bigr)$, we simulate $T$ steps of the discretized closed-loop dynamics
\begin{align}
\label{eqn:finetune-rollout}
    \hat\x_t & = \x_t + \mathbf{b}, \quad \x_{t+1} = \x_t + \Delta t \,(f(\x_t) + g(\x_t)\u_t) \nonumber\\
    \u_t & = \Pi_{\mathrm{NMR}}\bigl(\hat\x_t, \e, \xi\bigr),
\end{align}
where $\Pi_{\mathrm{NMR}}$ denotes the projection of the nominal control $\u_t^{\mathrm{nom}} = k(\hat\x_t)$ onto the NMR-tightened CBF half-space \cref{eqn:nmr}, clipped to the box $[\underline{\u},\overline{\u}]$ and with gradient $\partial \x_{t+1}/\partial \xi$.
% $\Pi_{\mathrm{NMR}}$ has a closed-form expression in the single-active-constraint case, obtained by selecting at each step the most violated constraint $i^\star = \arg\max_i d_i$ with deficits $d_i = a_i\u_t^{\mathrm{nom}} + L_fh_i(\hat\x_t) + \alpha(h_i(\hat\x_t)) + \rho_{\xi,i}$ and projecting against only that row.
% The $\arg\max$ contributes zero gradient, so reverse-mode autodiff flows cleanly through both the active row's $\rho_{\xi,i^\star}$ and $\hat\x_t$, enabling end-to-end gradients $\partial \x_{t+1}/\partial \xi$.

The per-episode loss then combines three terms:
\begin{align}
    \mathcal L(\xi)
    \;=\; &
    w_s \sum_{t=1}^{T} \bigl[h(\x_t) - \delta_{\mathrm{buf}}\bigr]_-^{2}
    \;+\; w_p \sum_{t=0}^{T-1} \bigl\|\u_t - k(\hat\x_t)\bigr\|_2^{2} \nonumber \\ 
    & \;+\; w_r \sum_{t=0}^{T-1} \bigl\|\rho_\xi(\hat\x_t, \e)\bigr\|_2^{2},
    \label{eqn:loss-finetune}
\end{align}
where $[z]_- = \max(0,-z)$.
The first term penalizes squared safety violations of the \emph{true} state against the  constraint $h(\x) \geq 0$, shifted by a small buffer $\delta_{\mathrm{buf}} > 0$ so that trajectories grazing $\partial\mathcal{C}$ incur loss before crossing.
The second term keeps the filtered control close to the nominal policy $k(\hat\x)$, mirroring the CBF-QP design objective $\min\|\u - \u^{\mathrm{nom}}\|^2$ from \cref{eqn:cbf} and providing a dense, low-variance shaping signal.
The regularizer pulls $\rho_\xi$ toward zero where neither safety nor minimal-deviation pressure prevents it, recovering nominal-CBF aggression at low $\e$ while preserving the softplus-enforced $\rho_\xi \geq 0$ structure.
% We bias the $\e$ sampling toward small radii (empirically, $\e = (\eta)^5 \overline\e$ with $\eta\sim\mathrm{Uniform}[0,1]$) so the fine-tune emphasizes the regime where \ana\ is most over-conservative and the learned residual has the most room to contract.

\subsubsection{Properties}
\label{sec:approach:nmrcbf:properties}

\iffalse
The NMR-CBF has two practical advantages over the \ana.
The first is runtime cost: $\rho_\xi$ is a single MLP evaluation, replacing the per-step projected-gradient inner loop of the \ana, which makes the NMR-CBF compatible with the control rates required by the hardware experiments of \cref{sec:hardware}.
The second is reduced conservativeness. By optimizing the closed-loop trajectory-level loss \cref{eqn:loss-finetune} rather than the pointwise worst case \cref{eqn:phi-oracle}, the fine-tuned network learns to assign smaller $\rho_\xi$ where the trajectory geometry or the box of $\uncs$ does not put the worst-case state on the active constraint.
This recovers nominal-CBF behavior in benign regimes while retaining \ana-like behavior near hard constraints.
The NMR-CBF replaces the DMR-CBF inner optimization with a single MLP evaluation and reduces conservativeness by optimizing the trajectory-level loss \cref{eqn:loss-finetune} rather than the pointwise worst case \cref{eqn:phi-oracle}.
The trade-off is a loss of formal guarantees away from the training distribution, since $\rho_\xi$ has no a priori lower bound on $\phi^\star$ and the NMR-CBF constraint can therefore under-tighten relative to \ana.
Despite this gap, the experiments in \cref{sec:results} show that the fine-tuned network is empirically as safe as the \ana.
\fi
The NMR-CBF replaces the DMR-CBF inner optimization with a single MLP evaluation and reduces conservativeness by optimizing the trajectory-level loss \cref{eqn:loss-finetune} rather than the pointwise worst case \cref{eqn:phi-oracle}.
The trade-off is that the learned robustness term has no formal guarantee away from the training distribution: $\rho_\xi$ has no a priori lower bound on $\phi^\star$, and the NMR-CBF constraint can therefore under-tighten relative to \ana.
Despite this gap, the experiments in \cref{sec:results} show that the fine-tuned network achieves empirical safety comparable to \ana.

%% file: results.tex
\section{Numerical Results}
\label{sec:results}
We evaluate safety, performance, scalability, and computational cost against existing measurement-robust filters.
%This section presents simulation results demonstrating various aspects of our proposed approach.
% In \cref{sec:results:dubins}, we compare the performance of the analytically calculated \ana\ to that of the Duality CBF and R-CBF-QP originally presented in \cref{sec:prelims:pointvsset} to show how our base approach addresses the issues posed by \cref{ex:dubins}.
%%In \cref{sec:results:dbint}, we consider safety and performance (with respect to \cref{eqn:cbf}) and compare our \ana\ and NMR-CBF approaches to a comprehensive set of baselines for a planar double integrator system.
In particular, we demonstrate the scalability of our proposed methods on a 12D quadrotor model that is subject to attitude and altitude constraints in addition to the obstacle avoidance constraint.
All experiments were conducted on a computer using Ubuntu 20.04 with 24 i9-10920X CPUs, two RTX 3090 GPUs, and 64 GB of RAM.

\subsection{Planar Double Integrator}
\label{sec:results:dbint}
Consider the system with state vector $\x = [p_x, p_y, v_x, \allowbreak v_y]^\top$ and dynamics
\begin{equation}
\label{eqn:double_integrator}
    \dot{\p} = \v \qquad \dot{\v} = \u,
\end{equation}
where $\p = [p_{x}, p_{y}]^\top$, $\v = [v_{x}, v_{y}]^\top$, and $\u = [u_{x}, u_{y}]^\top \in [-1, 1]^2$.
The safety filter must avoid an obstacle with radius \SI{0.25}{m} at $\p_{obs} = [0.0, 0.0]^\top$ despite a nominal controller $k_d(\hat{\x})$ designed to drive straight towards the goal position $\p_{g} = [2.0, 0.0]^\top$ and initial conditions sampled from a region $[-2.0, 0.0]^\top \oplus [0.5, 0.25]^\top$ opposite the obstacle.
To compare our \ana\ and NMR-CBF approaches to existing methods across varying levels of state estimation uncertainty, we consider \gdn~\cite{rober2026guardian}, CBF~\cite{ames2016control}, R-CBFs~\cite{nanayakkara_safety_2025} with ${\gamma}_1={\gamma}_2$ set to 0.2 (R-CBF (0.2)) and 1.0 (R-CBF (1.0)), the adaptive R-CBF-QP~\cite{das_safe_2025}, and the Duality CBF~\cite{tan2026duality} each subject to $\e_{\hat{x}} \triangleq [\varepsilon,\ \varepsilon,\ \frac{1}{2}\varepsilon,\ \frac{1}{2}\varepsilon]^\top$ with $\varepsilon \in \{0.0,\ 0.1,\ 0.2,\ 0.3,\ 0.4,\ 0.5\}$.
For each filtering method and each $\varepsilon$, a sample of $N_s = 1{,}000$ initial states $\x_0^i$ were collected with measurements subject to bias terms $\mathbf{b}_{\hat{x}}^i \sim \mathcal{U}(-\e_{\hat{x}}, \e_{\hat{x}})$ and rolled out according to their respective policies, again with $\hat{\x}^i = \x^i + \mathbf{b}_{\hat{x}}^i$.
As metrics for comparison, we consider collision rate and time to goal, the latter of which serves as a more interpretable proxy for \cref{eqn:cbf}, i.e., how much each filter deviates from the nominal control.

% \begin{equation}
%     \dot{x} = f(x) + G(x)\,u, \qquad
%     f(x) = \begin{bmatrix} v_x \\ 0 \\ v_y \\ 0 \end{bmatrix}, \quad
%     G(x) = \begin{bmatrix} 0 & 0 \\ 1 & 0 \\ 0 & 0 \\ 0 & 1 \end{bmatrix},
% \end{equation}

% \begin{equation}
%     x = (p_x, v_x, p_y, v_y), \qquad u = (u_x, u_y) \in [-1, 1]^2.
% \end{equation}

% \begin{equation}
%     h(x) = R_{\mathrm{obs}} - \sqrt{p_x^2 + p_y^2}, \qquad R_{\mathrm{obs}} = 0.25,
% \end{equation}

% \begin{equation}
%     \psi_1(x) = -\frac{p_x v_x + p_y v_y}{\sqrt{p_x^2 + p_y^2}}
%               + \alpha_0\bigl(R_{\mathrm{obs}} - \sqrt{p_x^2 + p_y^2}\bigr),
%     \qquad \alpha_0 = 2.
% \end{equation}

% \begin{equation}
% \begin{aligned}
%     \min_{u} \quad & \tfrac{1}{2}\|u - u_{\mathrm{nom}}\|^2 \\
%     \text{s.t.} \quad & L_f\psi_1 + L_G\psi_1\,u + \alpha_{\mathrm{QP}}\,\psi_1 \le 0, \\
%     & u \in [-1, 1]^2.
% \end{aligned}
% \end{equation}
The results for the sample initial states with different safety filtering approaches are shown visually in \cref{fig:dbint_tradeoff}. % and further quantified in \cref{tab:mc_outcome_rates_dbint}.
% We compare our \ana\ and NMR-CBF approaches against a set of baselines including \gdn\ \cref{eqn:gdn}, CBF~\cref{eqn:cbf_constraint}, R-CBFs~\cref{eqn:r-cbf} with $\bar{\gamma}_1=\bar{\gamma}_2$ set to 0.2 (R-CBF (0.2)) and 1.0 (R-CBF (1.0)), the adaptive R-CBF-QP~\cite{das_safe_2025}, and the Duality CBF~\cref{eqn:duality-cbf}.
For each policy and each $\varepsilon$, each of the $1{,}000$ sampled trajectories were recorded as one of \texttt{Reached}, \texttt{Timeout}, or \texttt{Unsafe}.
For trajectories categorized as \texttt{Reached}, the safety filter successfully avoided the obstacle and the state estimate $\hat{\x}$ reached the goal position within a \SI{20}{s} maximum simulation time limit.
For trajectories categorized as \texttt{Timeout}, the safety filter successfully kept $\x$ from collision, but did not allow the system to reach the goal position within the time limit.
Finally, for trajectories categorized as \texttt{Unsafe}, the safety filter allowed the system to collide with the obstacle.
In \cref{fig:dbint_tradeoff}, for each value of $\varepsilon$, the percentage of trajectories that reached the goal is labeled and colored according to the average time it took the system to reach the goal condition, the percent timeout trajectories is gray, and the percent unsafe trajectories is black.
These results are also visualized qualitatively for $\varepsilon = 0.3$ in \cref{fig:dbint_trajectories} where unsafe state trajectories are colored black and safe trajectories are colored to show the system's speed throughout each rollout.
% As shown in \cref{tab:mc_outcome_rates_dbint}, the nominal controller does not make any attempt at obstacle avoidance and has a high collision rate, so the safety of the system is left entirely to the safety filter.
\begin{figure*}[t]
    \vspace{4pt}
    \centering
    \includegraphics[width=0.98\linewidth, trim={0 0 0 6},clip]{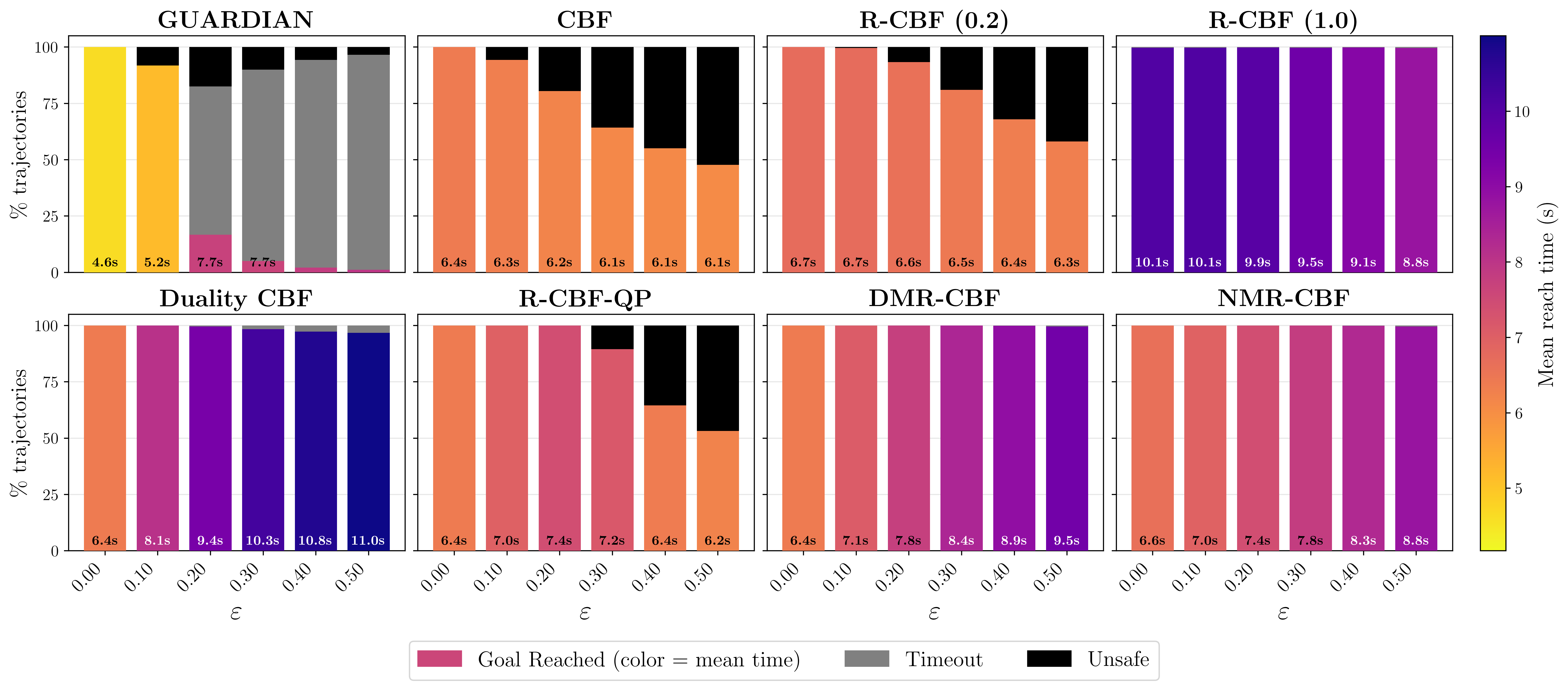}
    \caption[Bar plot of trajectory rollout results for a planar double integrator]{Results of trajectory rollouts show time to goal, safety rate, and timeout rate for each policy over a range of $\varepsilon$ values. 
    \ana\ and NMR-CBF are the only approaches that have no collisions and are not overly conservative.}
    \label{fig:dbint_tradeoff}
\end{figure*}

% \begin{table*}[t]
% \centering
% \scriptsize
% \caption[Reached$\,/\,$Slow$\,/\,$Unsafe outcome rates (\%) vs.\ uncertainty parameter $\varepsilon$ for planar integrator]{Reached$\,/\,$Slow$\,/\,$Unsafe outcome rates (\%) vs.\ uncertainty parameter $\varepsilon$}
% \label{tab:mc_outcome_rates_dbint}
% \begin{tabular}{llccccccccc}
% \toprule
% Category & Approach & \multicolumn{3}{c}{$\varepsilon=0.10$} & \multicolumn{3}{c}{$\varepsilon=0.30$} & \multicolumn{3}{c}{$\varepsilon=0.50$} \\
% \cmidrule(lr){3-5} \cmidrule(lr){6-8} \cmidrule(lr){9-11}
%  &  & Reached & Slow & Unsafe & Reached & Slow & Unsafe & Reached & Slow & Unsafe \\
% \midrule
% \multirow{2}{*}{Non-robust} & Nominal & 0.0 & 0.0 & 100.0 & 2.4 & 0.0 & 97.6 & 17.7 & 0.0 & 82.3 \\
%  & CBF & 94.3 & 0.0 & 5.7 & 64.3 & 0.0 & 35.7 & 47.7 & 0.0 & 52.3 \\
% \midrule
% \multirow{2}{*}{\shortstack[l]{Set-based\\robust}} & Duality CBF & 100.0 & 0.0 & 0.0 & 98.4 & 1.6 & 0.0 & 96.7 & 3.3 & 0.0 \\
%  & GUARDIAN & 88.9 & 2.9 & 8.2 & 3.1 & 86.9 & 10.0 & 1.0 & 95.6 & 3.4 \\
% \midrule
% \multirow{6}{*}{\shortstack[l]{Point-based\\robust}} & R-CBF ($\gamma{=}0.2$) & 99.6 & 0.0 & 0.4 & 81.0 & 0.0 & 19.0 & 58.1 & 0.0 & 41.9 \\
%  & R-CBF ($\gamma{=}1.0$) & 99.7 & 0.3 & 0.0 & 99.7 & 0.3 & 0.0 & 99.6 & 0.4 & 0.0 \\
%  & R-CBF-QP & 100.0 & 0.0 & 0.0 & 89.5 & 0.0 & 10.5 & 53.2 & 0.0 & 46.8 \\
%  & \ana & 100.0 & 0.0 & 0.0 & 100.0 & 0.0 & 0.0 & 99.6 & 0.4 & 0.0 \\
%  & NMR-CBF & 100.0 & 0.0 & 0.0 & 100.0 & 0.0 & 0.0 & 99.6 & 0.4 & 0.0 \\
% \bottomrule
% \end{tabular}
% \end{table*}

\begin{figure*}[th!]
    \centering
    % \captionsetup[subfigure]{aboveskip=2pt,belowskip=2pt}
    % \begin{subfigure}{\linewidth}
    %     \centering
    %     \includegraphics[width=0.85\linewidth]{figures/dbint_trajectories_0.1_r1.png}
    %     \caption{Trajectories with $\varepsilon = 0.1$ show all methods aside from CBF and \gdn\ maintian safety, R-CBF~(1.0) is very conservative.}
    %     \label{fig:dbint_trajectories:1}
    % \end{subfigure}
    % \begin{subfigure}{\linewidth}
    %     \centering
    \includegraphics[width=0.98\linewidth, trim={8pt 0 8pt 0},clip]{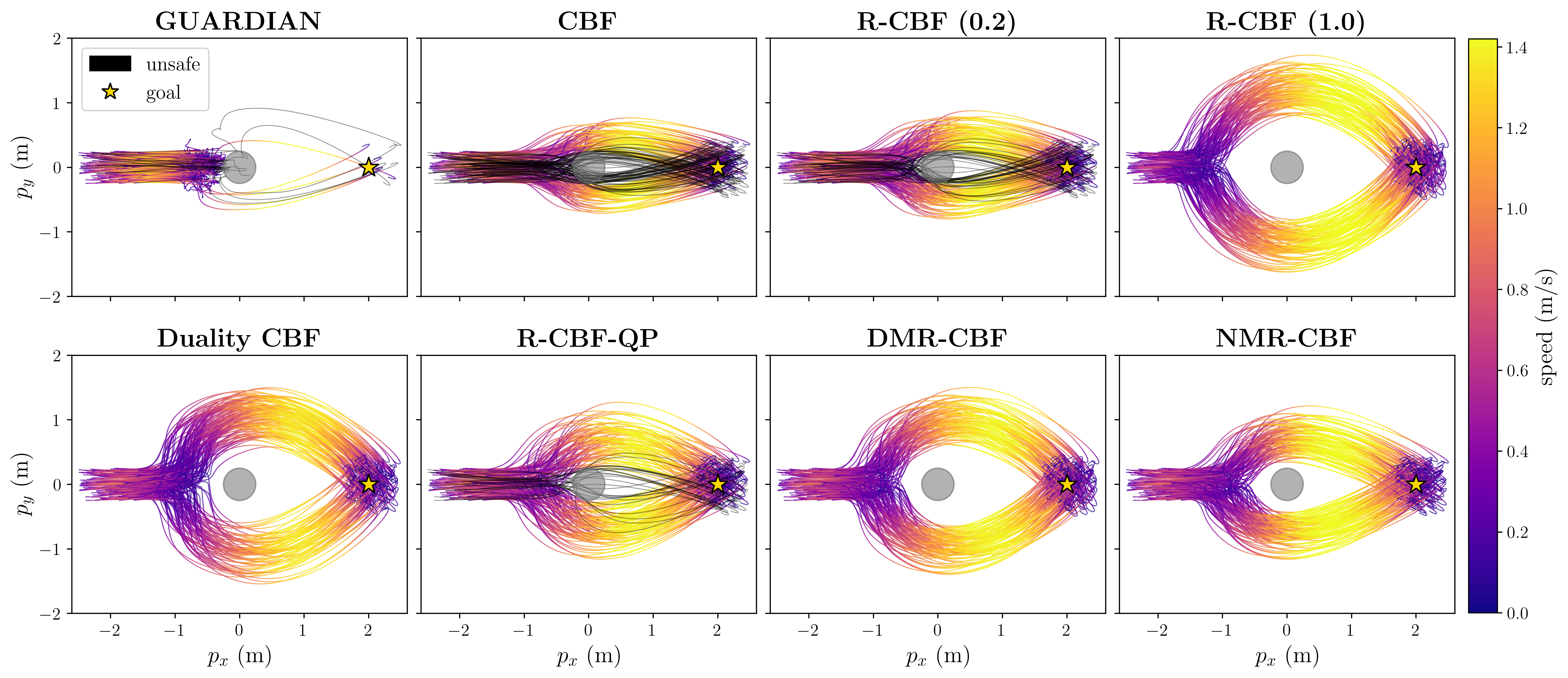}
        % \caption{Trajectories with $\varepsilon = 0.3$ show R-CBF~(0.2) and R-CBF-QP starting to fail more frequently.}
        % \label{fig:dbint_trajectories:3}
    % \end{subfigure}
    % \begin{subfigure}{\linewidth}
    %     \centering
    %     \includegraphics[width=0.85\linewidth]{figures/dbint_trajectories_0.5_r1.png}
    %     \caption{Trajectories with $\varepsilon=0.5$ show R-CBF~(0.2) and R-CBF-QP fail often; \ana\ and NMR-CBF remain safe but not as conservative as Duality CBF. 
    %     % \nr{TODO: clean up these figures a bit}
    %     }
    %     \label{fig:dbint_trajectories:5}
    % \end{subfigure}
    \caption[Sample trajectories with $\varepsilon = 0.3$ amounts of state-estimation uncertainty]{Sample trajectories with $\varepsilon = 0.3$ show R-CBF~(0.2) and R-CBF-QP starting to fail more frequently.}
    \label{fig:dbint_trajectories}
    \vspace{-12pt}
\end{figure*}

As shown by the two leftmost panels in \cref{fig:dbint_tradeoff}, the set-based Duality CBF and \gdn\ perform poorly with slower time to goal and more \texttt{Timeout} trajectories than the other approaches. % and the `Set-based robust' category in \cref{tab:mc_outcome_rates_dbint}
% For $\varepsilon = 0.0$, \gdn\ and Duality CBFs reduce to standard HJ reachability and CBF formulations, respectively, so perform well.
Despite good performance for $\varepsilon = 0.0$ (where \gdn\ and the Duality CBFs reduce to standard HJ reachability and CBF formulations, respectively), when $\varepsilon > 0$, these approaches are negatively impacted by the need to optimize over the entire set $\bar{\mathcal{X}}$, resulting in difficulty choosing which direction to go around the obstacle, as discussed in \cref{sec:prelims:prev} and \cref{sec:approach:dmrcbfs}.
In obstacle avoidance, \gdn's assumptions are not satisfied, resulting in safety violations and high timeout rates caused by an inability to decide which direction to go around the obstacle, as shown in \cref{fig:dbint_trajectories}.
The Duality CBF has similar difficulty getting around the obstacle, though not to the same degree, resulting in fewer timeouts, but high average time to goal.

The CBF, R-CBF~(0.2), and R-CBF~(1.0) each have different issues, but they all stem from the fact that these approaches do not adapt to the level of state estimation uncertainty.
The standard CBF allows an increasing number of collisions as $\varepsilon$ increases, which is unsurprising given that it does not account for state estimation error at all.
The R-CBF~(0.2) and R-CBF~(1.0) each have terms to increase their robustness, but it is not clear how to tune them to properly account for state-estimation error.
As a result, R-CBF~(0.2) is not robust to higher values of $\varepsilon$ and R-CBF~(1.0) is overly conservative, especially for low values of $\varepsilon$, resulting in slow time to goal.
Notably, the R-CBF~(1.0) approach has decreasing time to goal with increasing $\varepsilon$. 
This is because as $\varepsilon$ increases, the R-CBF~(1.0) intervenes less because the state estimates may be biased further away from the obstacle.

The R-CBF-QP performs well for low values of $\varepsilon$, recording no collisions for $\varepsilon = 0.0$, 0.1, and 0.2 with the best time to goal of any baseline.
However, as discussed in \cref{sec:prelims:prev}, its robustness guarantees deteriorate for higher levels of state uncertainty, so it fails more frequently with $\varepsilon = 0.3$, 0.4, and 0.5.
% \nr{TODO: Connect back to NFLs more thoroughly; we care about high uncertainty.}
The \ana\ has a perfect safety record and is less conservative than the Duality CBF and R-CBF~(1.0) (with the exception of R-CBF~(1.0) at $\varepsilon=0.5$), but reaches the goal slower than the R-CBF-QP for $\varepsilon = 0.1$ and 0.2.
Finally, though it loses the theoretical guarantees of the \ana, the NMR-CBF demonstrates the best performance of any approach: it has a perfect safety record and matches the time to goal of the R-CBF-QP for $\varepsilon = 0.1$, 0.2, and is at least as fast as R-CBF~(1.0) for $\varepsilon = 0.3$, 0.4, and 0.5.
Overall, our \ana\ and NMR-CBF approaches are the only approaches that successfully avoid the obstacle for all values of $\varepsilon$ without being overly conservative.

% \begin{figure}
%     \centering
%     \includegraphics[width=0.9\linewidth]{figures/dbint_trajectories_0.1_r1.png}
%     \caption{Caption}
%     \label{fig:trajectories_1}
% \end{figure}

% \begin{figure}
%     \centering
%     \includegraphics[width=0.9\linewidth]{figures/dbint_trajectories_0.3_r1.png}
%     \caption{Caption}
%     \label{fig:trajectories_3}
% \end{figure}

% \begin{figure}
%     \centering
%     \includegraphics[width=1.0\linewidth]{figures/dbint_trajectories_0.5_r1.png}
%     \caption{Caption}
%     \label{fig:trajectories_5}
% \end{figure}

\subsection{Quadrotor}
\label{sec:results:quadrotor}
Next, we investigate the scalability of our approaches with respect to state dimension by demonstrating them on a 12D quadrotor model with
\begin{equation}
\begin{split}
    \x = (p_x, p_y, p_z, v_x, v_y, v_z, \phi, \theta, \psi, p, q, r) \in \mathbb{R}^{12},
    \\ \u = (u_T, u_{\tau_x}, u_{\tau_y}, u_{\tau_z}) \in [-1, 1]^4,
\end{split}
\end{equation}
where $\p = [p_x, p_y, p_z]^\top$, $\v=[v_x, v_y, v_z]^\top$, $\phi$, $\theta$, and $\psi$ represent the roll, pitch, and yaw, respectively, and $p$, $q$, and $r$ represent their respective rates.
The inputs $u_T$, $u_{\tau_x}$, $u_{\tau_y}$, and $u_{\tau_z}$ represent the collective thrust and the body-frame torques about the roll ($x$), pitch ($y$), and yaw ($z$) axes, respectively.
% The full dynamics for the quadrotor are provided in \cref{appb:quadrotor_dynamics}.
As in \cref{sec:results:dbint}, the quadrotor must fly to a goal position while avoiding an obstacle with radius \SI{0.25}{m} positioned at the origin. 
In addition to the obstacle avoidance constraint, the quadrotor is subject to attitude constraints and floor/ceiling altitude constraints.
To ensure safety, a separate CBF condition is imposed for each constraint, amplifying the importance of the computational efficiency of each approach.

We repeat the evaluation of \cref{sec:results:dbint}, $N_s = 1{,}000$ trajectories, but \gdn\ is omitted because its HJ formulation does not scale to the 12D state space.
Otherwise, many of the same trends shown in \cref{fig:dbint_tradeoff} were observed, a snapshot of which are presented numerically in \cref{tab:mc_outcome_rates_quad3d}.
As shown in \cref{tab:mc_outcome_rates_quad3d}, each baseline safety filter allows more failures as $\varepsilon$ increases.
This includes the R-CBF (1.0), which is a notable difference from the results presented in \cref{sec:results:dbint}, demonstrating the challenge of manually setting R-CBF $\gamma$ values.
As in \cref{fig:dbint_tradeoff}, the NMR-CBF additionally matches or outperforms the other approaches in the time-to-goal metric, leaving our \ana\ and NMR-CBF as the only approaches with a perfect safety record without being overly conservative.
\begin{table}[t]
\vspace{4pt}
\centering
\scriptsize
\setlength{\tabcolsep}{4pt}
\caption{Time-to-goal (s) and Timeout$\,/\,$Unsafe outcome rates (\%) vs.\ noise parameter $\varepsilon$ on the 12D quadrotor.}
\label{tab:mc_outcome_rates_quad3d}
\begin{tabular}{@{}llcccccc@{}}
\toprule
Category & Approach & \multicolumn{3}{c}{$\varepsilon=0.10$} & \multicolumn{3}{c}{$\varepsilon=0.30$} \\
\cmidrule(lr){3-5} \cmidrule(lr){6-8}
 &  & Time & Timeout & Unsafe & Time & Timeout & Unsafe \\
\midrule
\multirow{2}{*}{Non-robust} & Nominal & --- & 0.0 & 100.0 & 4.33 & 0.0 & 81.9 \\
 & CBF & 4.50 & 0.0 & 0.1 & 4.50 & 0.0 & 17.7 \\
\midrule
\multirow{1}{*}{Set-based} & Duality CBF & 4.59 & 3.0 & 0.0 & 4.68 & 2.0 & 0.0 \\
\midrule
\multirow{6}{*}{Point-based} & R-CBF (0.2) & 4.50 & 0.0 & 0.0 & 4.51 & 0.0 & 15.5 \\
 & R-CBF (1.0) & 4.54 & 0.0 & 0.0 & 4.54 & 0.0 & 3.8 \\
 & R-CBF-QP & 4.50 & 0.0 & 0.0 & 4.51 & 0.0 & 15.7 \\
 & DMR-CBF & 4.53 & 0.0 & 0.0 & 4.58 & 0.0 & 0.0 \\
 % & NMR-CBF & 4.53 & 0.0 & 0.0 & 4.59 & 0.0 & 0.0 \\
 & NMR-CBF & 4.50 & 0.0 & 0.0 & 4.56 & 0.0 & 0.0 \\
\bottomrule
\end{tabular}
\ifshrink
\vspace{-12pt}
\fi
\end{table}

Moreover, as shown in \cref{tab:ms-per-step}, the NMR-CBF adds little computational overhead relative to a standard CBF, requiring 0.31 ms versus 0.26 ms per step for the planar DI and 0.57 ms versus 0.47 ms for the quadrotor. The DMR-CBF is more expensive due to its online inner optimization, but remains substantially faster than the Duality CBF, particularly for the 12D quadrotor (2.27 ms versus 45.28 ms per step).
\iffalse
Moreover, as shown in \cref{tab:ms-per-step}, both \ana s and NMR-CBFs have competitive computation times across both examples.
The NMR-CBF does not have much extra computational overhead when compared with the normal CBF since it is just a normal CBF with an inference call to the NN $\rho_{\xi}$.
The \ana\ is a bit slower since it must conduct the optimization \cref{eqn:amin_def} online, but it is still much faster than the Duality CBF, which must construct a polytopic over-approximation of the dual image set, requiring a separate inner optimization per polytope facet~\cite{tan2026duality}.
% Notably, \gdn\ is the fastest safety filtering approach since its derivative-free optimization does not require any iteration, but does not scale to 12D.
\fi
% \begin{table}[th]
%   \centering
%   \caption{Per-step computation time (ms/step).}
%   \label{tab:ms-per-step}
%   \begin{tabular}{lcc}
%     \toprule
%     Method        & Planar DI & Quadrotor \\
%     \midrule
%     Nominal       & 0.007 & 0.009  \\
%     CBF           & 0.263 & 0.471  \\
%     R-CBF (0.2)   & 0.264 & 0.505  \\
%     R-CBF (1.0)   & 0.311 & 0.588  \\
%     R-CBF-QP      & 0.300 & 0.544  \\
%     Duality CBF   & 2.730 & 45.282 \\
%     GUARDIAN      & 0.254 & ---    \\
%     \ana          & 0.609 & 2.266  \\
%     NMR-CBF       & 0.309 & 0.574  \\
%     \bottomrule
%   \end{tabular}
% \end{table}

\begin{table}[th]
  \centering
  \caption{Per-step computation time (ms/step).}
  \label{tab:ms-per-step}
  \begin{tabular}{lcc}
    \toprule
    Method        & Planar DI & Quadrotor \\
    \midrule
    Nominal       & 0.01 & 0.01  \\
    CBF           & 0.26 & 0.47  \\
    R-CBF (0.2)   & 0.26 & 0.51  \\
    R-CBF (1.0)   & 0.31 & 0.59  \\
    R-CBF-QP      & 0.30 & 0.54  \\
    Duality CBF   & 2.73 & 45.28 \\
    GUARDIAN      & 0.25 & ---    \\
    \ana          & 0.61 & 2.27  \\
    NMR-CBF       & 0.31 & 0.57  \\
    \bottomrule
  \end{tabular}
  \ifshrink
    \vspace{-12pt}
  \fi
\end{table}

% \begin{table}[t]
%   \centering
%   \caption{Per-step computation time (ms/step).}
%   \label{tab:ms-per-step}
%   \begin{tabular}{lrrr}
%     \toprule
%     Method       & Dubins Car & Planar DI & Quadrotor \\
%     \midrule
%     Nominal      & 0.021 & 0.007 & 0.009  \\
%     CBF          & 0.437 & 0.263 & 0.471  \\
%     R-CBF (0.2)  & 0.423 & 0.264 & 0.505  \\
%     R-CBF (1.0)  & 0.513 & 0.311 & 0.588  \\
%     R-CBF-QP     & 0.442 & 0.300 & 0.544  \\
%     Duality CBF  & 3.805 & 2.730 & 45.282 \\
%     GUARDIAN     & 0.075 & 0.254 & ---    \\
%     \ana         & 0.758 & 0.609 & 2.266  \\
%     NMR-CBF      & 0.494 & 0.309 & 0.574  \\
%     \bottomrule
%   \end{tabular}
% \end{table}

\section{Hardware Experiments: Quadruped}
\label{sec:hardware}
To validate our approach on hardware, we deployed it on a Unitree Go2 EDU (\cref{fig:go2}) equipped with an NVIDIA Jetson Orin Nano 8GB onboard computer, motor encoders, and an IMU.
We model the Go2 as a planar DI~\cref{eqn:double_integrator}, allowing us to use a formulation similar to \cref{sec:results:dbint} and apply the feedback linearization strategy discussed in \cite{d1992dynamic} to transform the DI inputs into commands for the Go2's \texttt{sportmode} API.
% We use a feedback linearization strategy modeling the the Go2 as a Dubins car, but controlled via a virtual control point modeled as a planar double integrator~\cref{eqn:double_integrator} with inputs $u_x,u_y$.
% Specifically, the feedback linearization approach discussed in \cite{d1992dynamic} is used to transform the virtual control point's $u_x,u_y$ inputs into speed and turning commands $v, \omega$ commands for a Dubins car model, which can be used to control the Go2 via Unitree's \texttt{sportmode} API.
\begin{figure}[tp]
    \vspace{4pt}
    \centering
    \includegraphics[width=0.8\linewidth]{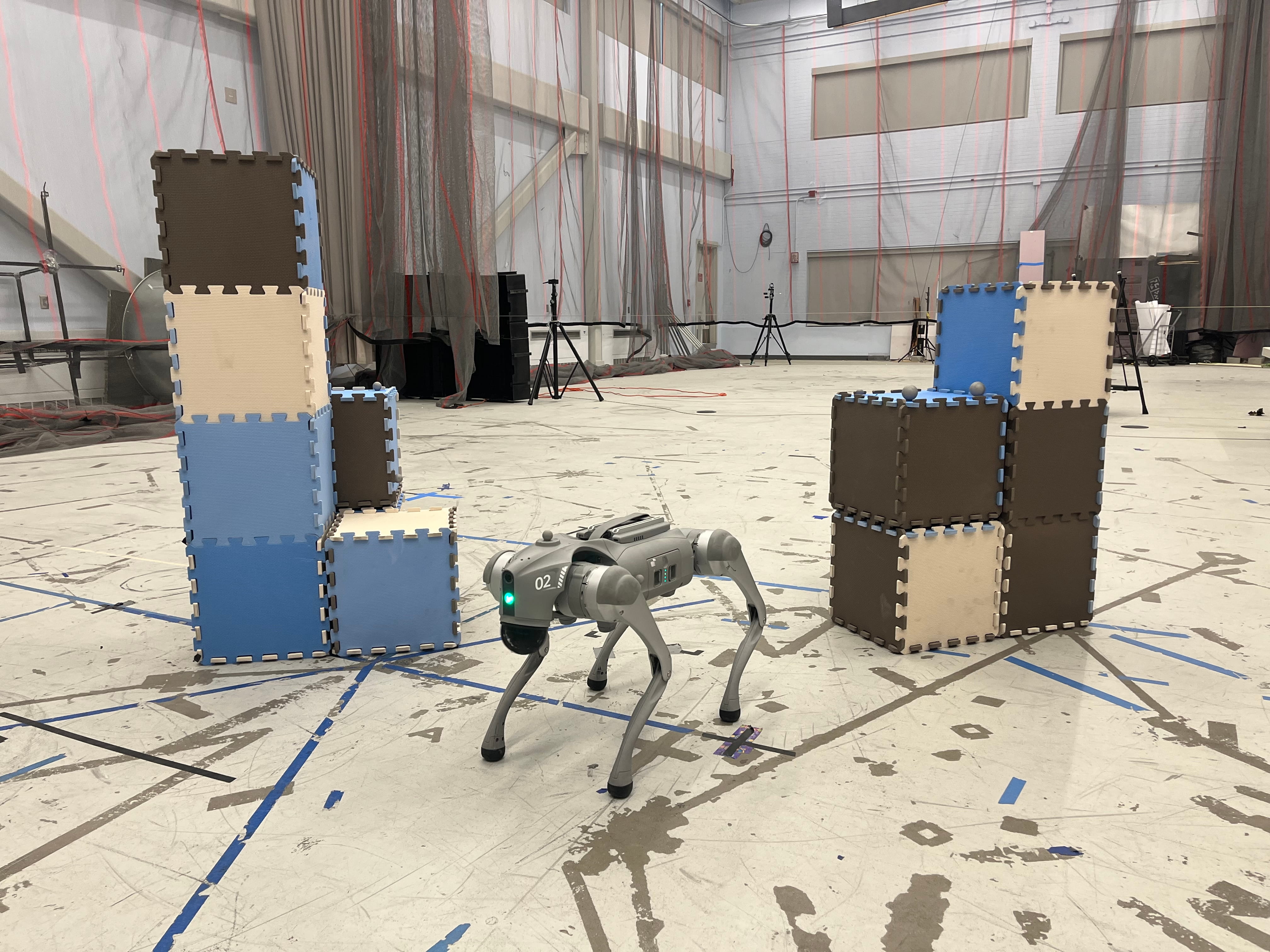}
    \caption{A Unitree Go2 in an experimental hardware space.}
    \label{fig:go2}
    % \vspace{-12pt}
\end{figure}

\begin{figure}
    \centering
    \includegraphics[width=0.98\linewidth]{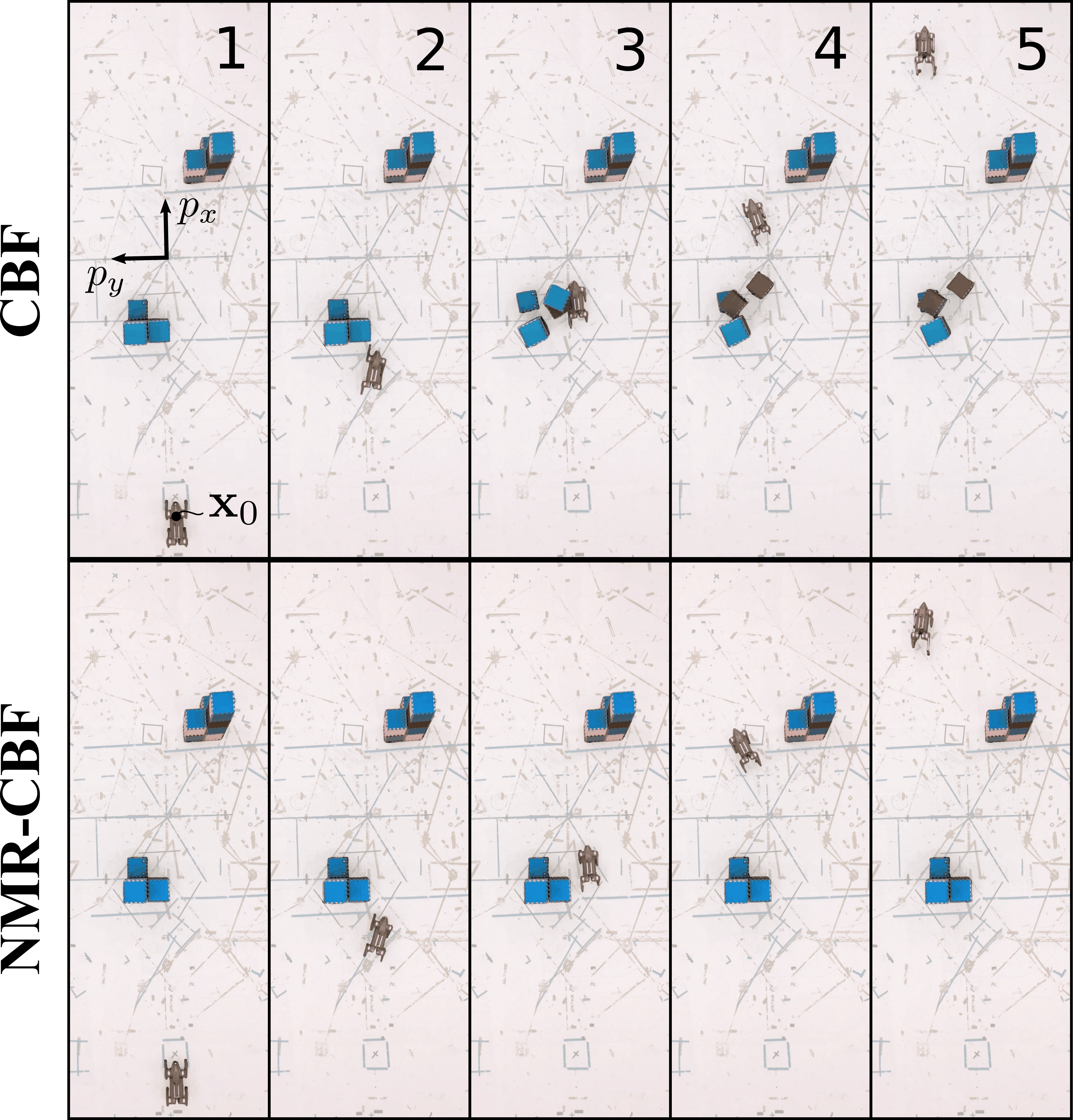}
    \caption{Hardware experiment comparing a standard CBF and NMR-CBF navigating obstacle field from similar initial conditions.
    The CBF does not account for state estimation error (artificially injected bias and odometry drift), resulting in collision with an obstacle.
    The NMR-CBF accounts for possible state estimation errors to avoid both obstacles.}
    \label{fig:go2_frames}
    \ifshrink
        \vspace{-16pt}
    \fi
\end{figure}

\begin{figure*}[tp]
    \centering
    \includegraphics[width=0.94\linewidth]{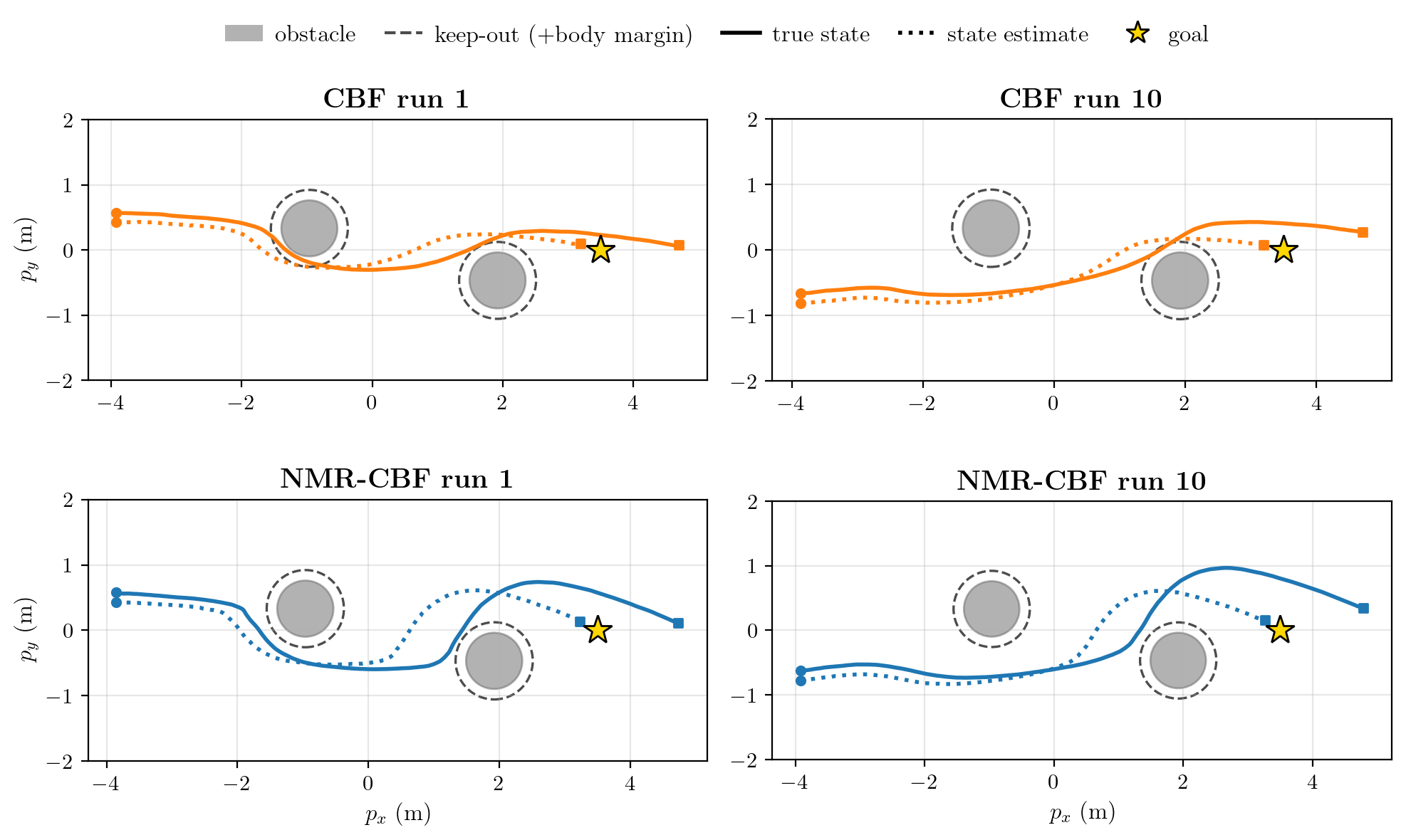}
    \caption[Go2 trajectories with CBF and NMR-CBF safety filters]{Go2 trajectories with CBF and NMR-CBF safety filters from different initial conditions.
    For the CBF (top), the state estimates (dotted orange) show the CBF allows the true state (solid orange) to collide with Obstacle 1 (top left) and Obstacle 2 (top right) depending on the initial configuration.
    The NMR-CBF (bottom) accounts for possible error in the state estimate, keeping the true state safe.}
    \label{fig:go2_trajectories}
    \ifshrink
        \vspace{-16pt}
    \fi
\end{figure*}

% {\color{red} use fig 6.9 from thesis? and can produce the corresponding five-frame NMR-CBF sequence from approximately the same initial condition? }

Our experimental setup consists of navigating from an initial position to a goal position while traversing a space with two obstacles, as originally presented in \cref{fig:go2_nmr_composite} and again in \cref{fig:go2_frames}.
% Ten trials, with initial conditions linearly spaced along a range of $p_y$ values, were conducted for each of the CBF and NMR-CBF safety filtering approaches.
At the beginning of each experiment, the ground-truth position of the Go2 was recorded using a Vicon motion-capture system and used to calibrate the Go2's estimate of the obstacles' relative position.
The Go2 was then commanded to drive towards the goal position using only its onboard state estimate, which is produced by the robot's proprioceptive odometry estimate and is not corrected by any external positioning system.
As a result, any drift in the Go2's state estimate introduces error in its position relative to the true position of the obstacles, leaving it susceptible to collision.
In test runs prior to our experiments, this drift was found to be significant, with the Go2 consistently underestimating the distance it traveled by up to 18\%.
In addition to error from drift, we introduced a negative \SI{15}{cm} bias into the Go2's estimate of $p_y$, meaning its state estimate is right of the true state.
As shown by the time series in \cref{fig:go2_frames}, these errors in the Go2's state estimate caused a standard CBF to allow collision with obstacles, while the NMR-CBF successfully protected the system. 

% We present our results through three complementary views.
% \cref{fig:go2_cbf_series,nmrcbfs:fig:go2_nmr_composite} provide visual intuition for two outcomes: \cref{fig:go2_cbf_series} shows a representative CBF run in which the Go2 collides with Obstacle 1, while \cref{fig:go2_nmr_composite} shows the NMR-CBF safely traversing the same configuration, overlaid with a sketch of the true and estimated state
% trajectories.
% These demonstrations were recorded separately from the quantitative trials and are illustrative; the sketched trajectories in \cref{fig:go2_nmr_composite}, in particular, are not measured data. 

% \begin{figure}[tp]
%     \centering
%     \includegraphics[width=1.0\linewidth]{figures/go2_cbf_r1.png}
%     \caption[Time series shows a CBF fails to protect the Go2 with state estimation error]{Time series shows a CBF fails to protect the Go2 from collision.
%     The presence of a negative bias in the estimate of $p_y$ causes the CBF to act as if the robot is further to the right, thus issuing an insufficient corrective action and colliding with Obstacle 1.}
%     \label{fig:go2_cbf_series}
% \end{figure}

% \begin{figure}[tp]
%     \centering
%     \includegraphics[width=1.0\linewidth]{figures/go2_nmr_r3.png}
%     \caption[Composite image of the Go2 using the NMR-CBF]{A composite image shows the NMR-CBF safely navigating the obstacles, despite the same estimation errors present in \cref{fig:go2_cbf_series}.}
%     \label{fig:go2_nmr_composite}
% \end{figure}

% The results of the ten trials per method are summarized in \cref{tab:go2_safety}.
To thoroughly evaluate our approach, ten trials, with initial conditions linearly spaced along a range of $p_y$ values, were conducted for each of the CBF and NMR-CBF safety filtering approaches.
\cref{tab:go2_safety} summarizes the collisions and safety rates across all ten trials per method and  \cref{fig:go2_trajectories} shows the measured ground-truth and estimated trajectories for the leftmost and rightmost initial conditions in the $p_y$ range under both the CBF and NMR-CBF filters.
Across all trials, the Go2's onboard estimate both lags and sits to the right of its true position: behind due to odometry drift (the Go2 under-travels by up to 18\%), and to the right due to the injected \SI{15}{cm} bias in $p_y$.
Because the CBF trusts this estimate directly, it mislocates the obstacles relative to the robot and, in several trials, fails to issue a sufficient correction.
Two distinct failure modes appear, both visible in the top row of \cref{fig:go2_trajectories}.
For initial conditions on the left of the obstacle configuration (e.g., CBF run~1, top-left panel), the rightward bias places the estimate on the far side of the first obstacle from the true state, so the CBF under-corrects and the ground-truth state enters the keep-out region (the obstacle inflated to account for the Go2's body) and collides with the first obstacle.
For initial conditions on the right (e.g., CBF run~10, top-right panel), the rightward bias alone would steer the robot clear of the second obstacle; however, the large along-track drift means the Go2 is much further along its path than its estimate indicates, and it collides with the second obstacle.

\begin{table}[t]
\vspace{4pt}
\centering
\caption{Collisions and safety rate over 10 hardware trials per method.
A collision is recorded when the ground-truth robot center crosses the keep-out boundary.}
\label{tab:go2_safety}
\begin{tabular}{lccc}
    \toprule
    & \multicolumn{2}{c}{Collisions} & \\
    \cmidrule(lr){2-3}
    Method & Obstacle 1 & Obstacle 2 & Safety rate \\
    \midrule
    CBF     & 1 & 2 & 7/10 \\
    NMR-CBF & 0 & 0 & 10/10 \\
    \bottomrule
\end{tabular}
\ifshrink
    \vspace{-18pt}
\fi
\end{table}

In contrast, the NMR-CBF remains safe across all initial conditions, as shown in \cref{fig:go2_nmr_composite} and the bottom row of \cref{fig:go2_trajectories}.
By considering $\e_{\hat{x}}$, which was configured to dynamically cover both the injected bias and the drift error as it grows over each trial, the NMR-CBF's learned robustness term properly accounts for the error in the state estimate.
As shown in the bottom left panel of \cref{fig:go2_trajectories}, the NMR-CBF allows the state estimate to get closer to Obstacle 1 than Obstacle 2 since there is less drift accumulated early in the trajectory.
The NMR-CBF therefore successfully keeps the ground-truth state away from both obstacles, despite the presence of the same estimation errors that cause the CBF to fail.

%% file: conclusion.tex
\section{Conclusion}
\label{sec:conclusion}
This paper presented two CBF-based formulations for safe obstacle avoidance under state estimation error.
First, the DMR-CBF provides a CBF formulation with an
{\color{internal}\emph{a posteriori}} safety certificate, but requires
an inner optimization and is more conservative than our second
approach: the NMR-CBF.
The NMR-CBF uses data generated by the DMR-CBF to pretrain a learned robustness term that is then fine-tuned using trajectory rollouts.
Numerical results show that both the DMR-CBF and NMR-CBF offer significant {\color{internal}empirical} safety improvements over all baselines, especially in high-uncertainty regimes.
Hardware experiments on a quadruped robot further demonstrate the practical utility of NMR-CBFs, protecting the system against the odometry drift naturally present in the onboard state estimation module.